\documentclass[10pt,aps,prl,twocolumn,superscriptaddress,floatfix]{revtex4-2}
\pdfoutput=1
\usepackage[utf8]{inputenc}
\usepackage[english]{babel}
\usepackage[T1]{fontenc}
\usepackage{amsmath}
\usepackage{amsfonts}
\usepackage{amssymb}
\usepackage{amsthm}
\usepackage{makeidx}
\usepackage{graphicx}
\usepackage{color}
\usepackage{xspace}
\usepackage{dsfont}
\usepackage{tabularray}
\usepackage{xspace}
\usepackage{float}
\definecolor{mpBlue}{RGB}{21, 101, 192}
\definecolor{mpGreen}{RGB}{46, 125, 50}
\usepackage{hyperref}
\hypersetup{citecolor=mpGreen}
\hypersetup{colorlinks=true}
\hypersetup{linkcolor=mpBlue}
\hypersetup{urlcolor=mpBlue}
\newtheorem{theorem}{Theorem}
\newtheorem*{theorem*}{Theorem}

\newtheorem{definition}[theorem]{Definition}

\renewcommand{\paragraph}[1]{\addcontentsline{toc}{section}{#1}\emph{#1.}---}
\makeatletter
\newcommand*{\bra}[1]{\langle #1 \@ifnextchar\ket{}{|}}
\newcommand*{\ket}[1]{| #1 \rangle \@ifnextchar\bra{\!}{}}
\makeatother
\newcommand*{\ketbra}[1]{| #1 \rangle \! \langle #1 |}

\DeclareMathOperator{\e}{e}
\DeclareMathOperator{\Ha}{\mathcal{H}}
\DeclareMathOperator{\bound}{\mathcal{B}}
\DeclareMathOperator{\I}{\mathds{1}}
\DeclareMathOperator{\Tr}{Tr}
\DeclareMathOperator{\state}{\mathcal{S}}
\DeclareMathOperator{\effect}{\mathcal{E}}
\newcommand*{\pol}{\mathrm{pol}}
\newcommand*{\pth}{\mathrm{path}}
\newcommand*{\cglmp}{CGLMP\textsubscript{4}\xspace}

\begin{document}

\title{Experimental implementation of dimension-dependent contextuality inequality}
\author{Emil H\r{a}kansson}
\affiliation{Institute for Quantum Information Science and Department of Physics \& Astronomy, Stockholm University, Stockholm, Sweden}
\affiliation{Hitachi Energy AB V\"{a}ster\r{a}s}
\author{Amelie Piveteau}
\affiliation{Institute for Quantum Information Science and Department of Physics \& Astronomy, Stockholm University, Stockholm, Sweden}
\author{Alban Seguinard}
\affiliation{Institute for Quantum Information Science and Department of Physics \& Astronomy, Stockholm University, Stockholm, Sweden}
\author{Muhammad Sadiq}
\affiliation{Institute for Quantum Information Science and Department of Physics \& Astronomy, Stockholm University, Stockholm, Sweden}
\author{Mohamed Bourennane}
\affiliation{Institute for Quantum Information Science and Department of Physics \& Astronomy, Stockholm University, Stockholm, Sweden}

\author{Otfried G\"{u}hne}
\affiliation{Naturwissenschaftlich-Technische Fakult\"{a}t, Universit\"{a}t Siegen, Walter-Flex-Stra\ss e 3, 57068 Siegen, Germany}

\author{Martin Pl\'{a}vala}
\affiliation{Institut f\"{u}r Theoretische Physik, Leibniz Universit\"{a}t Hannover, Hannover, Germany}
\affiliation{Naturwissenschaftlich-Technische Fakult\"{a}t, Universit\"{a}t Siegen, Walter-Flex-Stra\ss e 3, 57068 Siegen, Germany}

\begin{abstract}
We present a derivation and experimental implementation of a dimension-dependent contextuality inequality to certify both the quantumness and dimensionality of a given system. Existing methods for certification of the dimension of quantum system can be cheated by using larger classical systems, creating a potential loophole in these benchmarks. Our approach uses contextuality inequalities that cannot be violated by classical systems thus closing the previous loophole. We validate this framework experimentally with photons, observing violations of a CHSH-based contextuality inequality and surpassing the qutrit bound of the \cglmp{}-based contextuality inequality. These show that contextuality can be used for noise-robust tests of the number of qubits.
\end{abstract}

\maketitle

\paragraph{Introduction}%
Understanding and harnessing the non-classical features of quantum systems has been a central focus of quantum information for decades. These non-classical features have many forms, such as entanglement \cite{guhne2009entanglement}, Bell nonlocality \cite{rosset2014classifying,brunner2014bell}, or contextuality \cite{kochen1990problem,budroni2022kochen,spekkens2005contextuality,muller2023testing} and experimental proofs certifying the presence of these non-classical features had significant impact on our understanding of physics. In order to progress towards practical applications of quantum information, it is necessary to not only experimentally verify the fundamental non-classical features of experimental systems at hand, but we must also develop meaningful benchmarks that test the quantum devices that are currently available. One of the most basic properties of quantum computer to benchmark is the number of qubits present in the quantum computer, effectively certifying the dimension of the underlying Hilbert space. Several methods for certifying the dimension of quantum systems were previously developed \cite{brunner2008testing,gallego2010device,brunner2013dimension,bowles2014certifying,tavakoli2015quantum,spee2020genuine}, but these share a common flaw: one can either pass the test by having quantum system of sufficient size, but one can pass the test also by having classical system of (usually larger) sufficient size. In other words, these test would be passed also by classical computers instead of quantum computers, thus opening a loophole in benchmarking the number of qubits within a quantum computer.

In this paper we close this loophole by constructing contextuality inequalities with dimension-dependent maximal violation. Since contextuality inequalities hold, under some assumptions, for classical systems, any violation proves that we are accessing a quantum system; observing a sufficiently high violation then proves that we are accessing a quantum system of at least certain dimension. Since our approach is based on linear contextuality inequalities, it is by construction noise-robust and thus experimentally applicable. Following the theoretical development of our dimension-dependent contextuality inequality, we demonstrate its usefulness and validate our findings experimentally through an implementation using photons. We observe a violation of a basic contextuality inequality based on the CHSH inequality already proposed in \cite{plavala2024contextuality}, and we observe violations of the qutrit bound of the \cglmp contextuality inequality in two separate experiments.

\paragraph{Theoretic background}%
The experimental setups we will consider in this work can be described in terms of preparations and measurements. Within the framework of quantum theory, preparations are described by density matrices and measurements by positive operator valued measures (POVMs). Let $\Ha$ be a finite-dimensional complex Hilbert space, let $\bound(\Ha)$ denote the set of linear operators on $\Ha$, we will use $\I$ to denote the identity operator, and for $A \in \bound(\Ha)$ we will use $A \geq 0$ to denote that $A$ is positive semi-definite and $\Tr(A)$ the trace of $A$. Then a density matrix $\rho$ is a positive semi-definite operator with unit trace, $\rho \geq 0$ and $\Tr(\rho) = 1$, and POVM $\{M_a\}$ is a set of positive semi-definite operators summing to the identity, that is, $M_a \geq 0$ and $\sum_a M_a = \I$. Given a density matrix $\rho$ and POVM $\{M_a\}$ the probability distribution that we obtain by measuring $\rho$ with $\{M_a\}$ is given according to the Born rule as $p(a) = \Tr(\rho M_a)$.

The question we aim to answer is the following: given access to the experimental data, we want to know whether the experimental data can be explained using a classical model, the so-called hidden variable model, and if not, then what is the minimal dimension of the Hilbert space $\Ha_{\min}$ such that some preparations and POVMs yield the same experimental data. We will introduce preparation and measurement (P\&M) contextuality \cite{spekkens2005contextuality,schmid2021characterization,schmid2024addressing} to answer the first question, and we will show that the second question can be answered using the very same methods.
\begin{definition}
Let $\state$ be the set of allowed preparations and let $\effect$ be the set of allowed measurements, then $(\state, \effect)$ is P\&M noncontextual if there are sets of operators $\{N_\lambda\}_{\lambda \in \Lambda}$ and $\{\omega_\lambda\}_{\lambda \in \Lambda}$ such that for all $\sigma \in \state$, $M \in \effect$ we have
\begin{equation} \label{eq:PMcontextualityTr}
\Tr(\sigma M) = \sum_\lambda \Tr(\sigma N_\lambda) \Tr(\omega_\lambda M)
\end{equation}
and
\begin{align}
&\Tr(\sigma N_\lambda) \geq 0,
&& \sum_\lambda \Tr(\sigma N_\lambda) = 1, \\
&\Tr(\omega_\lambda M) \geq 0,
&& \Tr(\omega_\lambda) = 1.
\end{align}
\end{definition}
The P\&M contextuality captures whether the allowed preparation $\sigma$ can be explained in terms of a probability distribution $p(\lambda|\sigma) = \Tr(\sigma N_\lambda)$ and, at the same time, whether the allowed measurement outcome $M$ can be explained as response functions $p(M|\lambda) = \Tr(\omega_\lambda M)$ in such way that they correctly reconstruct the predictions of quantum theory. In this way one obtains a hidden variable model in terms of the classical variable $\lambda$ that is well defined and operationally motivated. A detailed discussion on
the underlying assumptions can be found in \cite{tavakoli2020incompatibility, guehne2023incompatibility}
and the relation to other notions of contextuality is discussed in \cite{budroni2022kochen}.

In this work we will consider set of subnormalized preparations $\{\sigma_{a|x}\}$ such that
\begin{equation} \label{eq:assemblageConstraint}
\sum_a \sigma_{a|x} = \sigma_*, \forall x \quad \text{where} \quad \Tr(\sigma_*) = 1.
\end{equation}
For a fixed value of $x$, $\sigma_{a|x}$ can be seen as a series of sub-normalized preparations adding up to a quantum state. We will also consider a set of POVMs $\{M_{b|y}\}_{b,y}$ and $\sum_b M_{b|y} = \I$ for all $y$. It is straightforward to extend the definition of P\&M contextuality to include sub-normalized preparations as they are just positive multiples of allowed preparations. Using the recently developed methods \cite{plavala2024contextuality,wright2023invertible} we obtain the following:
\begin{theorem} \label{thm:PMthenBell}
Consider the sub-normalized preparations $\{\sigma_{a|x}\}_{a,x}$ obeying \eqref{eq:assemblageConstraint} and a set of POVMs $\{M_{b|y}\}_{b,y}$. Let $p(ab|xy) = \Tr( \sigma_{a|x} M_{b|y})$ denote the conditional probability one obtains from the experiment, conditioned on the choice of $x$ and $y$. If $(\{\sigma_{a|x}\}_{a,x}, \{M_{b|y}\}_{b,y})$ is P\&M noncontextual, then
\begin{equation} \label{eq:BellLocal}
p(ab|xy) = \sum_\lambda p(\lambda) p(a|x,\lambda) p(b|y,\lambda)
\end{equation}
where $p(\lambda) \geq 0$, $\sum_\lambda p(\lambda) = 1$, $p(a|x,\lambda) \geq 0$, $\sum_a p(a|x,\lambda) = 1$, $p(b|y,\lambda) \geq 0$, $\sum_b p(b|y,\lambda) = 1$.
\end{theorem}
The result is obtained by taking $p(\lambda) p(a|x,\lambda) = \Tr(\sigma_{a|x} N_\lambda)$, and $p(b|y,\lambda) = \Tr(\omega_\lambda M_{b|y})$, the complete proof is relegated to the appendix. Thm.~\ref{thm:PMthenBell} shows that we can use Bell inequalities \cite{bell1964einstein,brunner2014bell,rosset2014classifying} to construct contextuality inequalities: if $p(ab|xy)$ is observed as no-signaling behavior in a Bell experiment, then it is Bell local exactly if it satisfies \eqref{eq:BellLocal}, and if such $p(ab|xy)$ does not satisfy \eqref{eq:BellLocal}, then there is a Bell inequality that it violates. 

Moreover, we can further show that one can also transport the respective dimension-dependent bounds. To see this, let $W$ be a Bell inequality and let $W([p(ab|xy)])$ denote the value of the Bell inequality on the no-signaling behavior $p(ab|xy)$, we will assume that $W([p(ab|xy)]) \leq 0$ if $p(ab|xy)$ is Bell local, i.e., if it satisfies \eqref{eq:BellLocal}, we can assume this without loss of generality as we can always rescale and offset any Bell inequality. Let us assume that there is a dimension-dependent bound $X_d$ for $W$, that is, there is a number $X_d$ such that for $\dim(\Ha_d) = d$, for any state $\rho \in \bound(\Ha_d \otimes \Ha_d)$ and for any sets of POVMs $R_{a|x} \in \bound(\Ha_d)$, $S_{b|y}  \in \bound(\Ha_d)$ we have $W([\Tr(\rho(R_{a|x} \otimes S_{b|y}))]) \leq X_d$. We can now show that the same dimension-dependent bound will hold when $W$ will be considered as contextuality inequality.
\begin{theorem} \label{thm:PMdim}
Let $W$ be a Bell inequality with dimension-dependent bound $X_d$. Let $\dim(\Ha_d) = d$ and let $\sigma_{a|x} \in \bound(\Ha_d)$ be sub-normalized preparations obeying \eqref{eq:assemblageConstraint} and let $M_{b|y} \in \bound(\Ha_d)$ be a set of POVMs. Then
\begin{equation}
W([\Tr(\sigma_{a|x} M_{b|y}]) \leq X_d.
\end{equation}
\end{theorem}
The proof follows by using the GHJW theorem \cite{schrodinger1935discussion,hadjisavvas1981properties,gisin1989stochastic,hughston1993complete,sainz2015postquantum} which states that a bipartite no-signaling assemblage can be prepared by measuring an entangled state, i.e., there are a quantum state $\rho \in \bound(\Ha_d \otimes \Ha_d)$ and set of POVMs $N_{a|x} \in \bound(\Ha_d)$ such that $\sigma_{a|x} = \Tr_A (\rho (N_{a|x} \otimes \I))$. The complete proof is relegated to the appendix.

\paragraph{CHSH contextuality inequality}%
The case of CHSH inequality \cite{clauser1969proposed,clauser1970proposed} was already investigated in \cite{plavala2024contextuality}, and we just shortly reviewed the results. For $a,b,x,y \in \{0,1\}$ and for $\sigma_{a|x}$ satisfying \eqref{eq:assemblageConstraint} we have that if $(\{\sigma_{a|x}\}_{a,x}, \{M_{b|y}\}_{b,y})$ is P\&M noncontextual, then according to Thm.~\ref{thm:PMthenBell} $p(ab|xy) = \Tr(\sigma_{a|x} M_{b|y})$ satisfies the CHSH inequality, which coincides with \cite[Proposition 5.]{plavala2024contextuality}.

\paragraph{\cglmp contextuality inequality}%
We will now apply the results of Thms.~\ref{thm:PMthenBell} and \ref{thm:PMdim} to the \cglmp Bell inequality \cite{collins2002bell,collins2004relevant}, and we will use the results of \cite{moroder2013device,cai2016new} to obtain dimension bounds as well. Let $a,b \in \{0, 1, 2, 3\}$ and $x,y \in \{0, 1\}$, the the \cglmp Bell inequality reads
\begin{equation}
\begin{split}
I_4([p(ab|xy)] &= p(a \leq b|00) + p(a \geq b|01) \\
&+ p(a \geq b|10) - p(a \geq b|11) - 2 \leq 0,
\end{split}
\end{equation}
where $p(a \leq b|00)$ is the sum of all of the terms $p(ab|00)$ where $a \leq b$, and analogically for the other terms. The qubit bound $X_2$ for the \cglmp is $X_2 = \frac{1}{\sqrt{2}} - \frac{1}{2} \approx 0.2071$ \cite{moroder2013device} and $X_3 = 0.315$ \cite{cai2016new}. The maximal violation of $0.36476$ is given by the following preparations and measurements \cite{cai2016new}: let $\dim(\Ha) = 4$, then
\begin{align}
&\sigma_{a|x} = \dfrac{\ketbra{\varphi_{a|x}}}{2}, \,
\ket{\varphi_{a|x}} = \dfrac{\sum_{k=0}^3 \e^{i(\frac{\pi}{2} a - \theta_x) k} s_k \ket{k}}{2} \label{eq:CGLMP4optimalAssemblage} \\
&M_{b|y} = \ketbra{\psi_{b|y}}, \;
\ket{\psi_{b|y}} = \dfrac{\sum_{k=0}^3 \e^{i(\frac{\pi}{2} b + \phi_y) k} \ket{k}}{2} \label{eq:CGLMP4optimalMeasurement}
\end{align}
where $\ket{k}$ is the standard computational basis, $\theta_0 = -\frac{\pi}{8}$, $\theta_1 = \frac{\pi}{8}$, $s_0 = s_3 = \cos(\vartheta)$, $s_1 = s_2 = \sin(\vartheta)$, $\vartheta =\frac{1}{2} \arcsin(\frac{1+2\sin(\pi / 8))}{\sqrt{6\sin^2(\pi / 8) + 4\sin(\pi / 8) + 1}})$, and $\phi_0 = 0$, $\phi_1 = \frac{\pi}{4}$.

In the experimental implementation with photos, we will implement the four level system needed to measure the maximal violation by using the polarization degree of freedom of the photon described by the Hilbert space $\Ha_{\pol}$, and by using an interferometer with two possible paths we get the which-path degree of freedom, which is described by the Hilbert space $\Ha_{\pth}$. Since $\dim(\Ha_\pth) = \dim(\Ha_\pol) = 2$, the Hilbert space describing the state of the photon $\Ha = \Ha_\pol \otimes \Ha_\pth$ has $\dim(\Ha) = 4$ as needed. As one can check, the preparations in \eqref{eq:CGLMP4optimalAssemblage} are entangled when we consider the factorization $\Ha = \Ha_\pol \otimes \Ha_\pth$ and thus they may be challenging to implement in the described interferometric setup. A more "experimental-friendly way" of doing the experiment would be with only separable preparations. For this reason we will also investigate what is the maximal violation of the \cglmp contextuality inequality with only separable preparations. We will use a see-saw algorithm: if we fix the measurements, optimizing the violation of the inequality over the separable preparations can be done in a single semi-definite program \cite{boyd2004convex}, here we are using that, in this case, separability coincides with positive partial transpose \cite{horodecki1996separability,peres1996separability,fawzi2021set}. If we fix separable preparations and optimize over measurements we again get a single semi-definite program. We can thus start with random measurements and alternately optimize over preparations and measurements to obtain the highest possible violation. The details of the algorithm are in the appendix, and the code and generated data are available \cite{plavala_github}. The best found violation was $0.33609$, which is obtained by using the same measurements as in \eqref{eq:CGLMP4optimalMeasurement} and the following preparations:
\begin{equation} \label{eq:CGLMP4separableAssemblage}
\begin{split}
&\sigma_{a|x} = \dfrac{1}{4} \ketbra{u_{a|x}} \otimes \ketbra{v_{a|x}}, \\
&\ket{u_{a|x}} = \dfrac{\e^{i((-1)^x \frac{\pi}{8} + \frac{\pi}{2} a)} \ket{0} + \e^{i((-1)^x \frac{3\pi}{8} - \frac{\pi}{2} a)} \ket{1}}{\sqrt{2}}, \\
&\ket{v_{a|x}} = \dfrac{\e^{-i((-1)^x \frac{\pi}{8} + \frac{\pi}{2} a)} \ket{0} + \ket{1}}{\sqrt{2}}.
\end{split}
\end{equation}

\begin{figure*}
\includegraphics[width=0.7\linewidth]{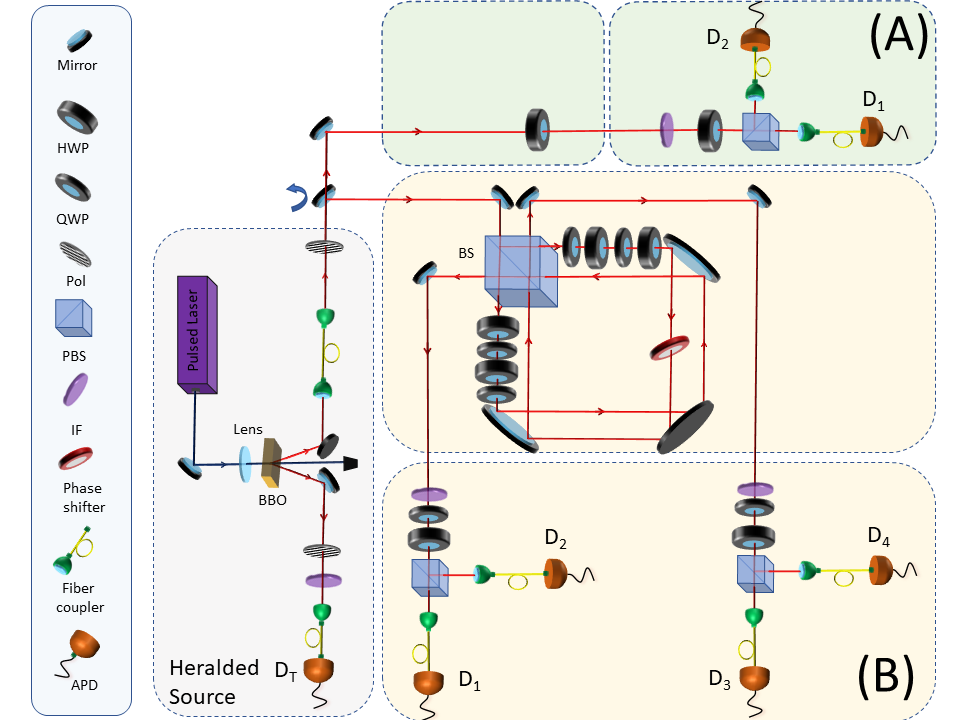}
\caption[Experimental setup]{{\bf Experimental setup with heralded source.} Initial state preparation using photons from a heralded source and a polarizer. The movable mirror either reflects photons to the experimental setup (A) or (B). The CHSH experiment (A) is realized by using HWP for state preparation and HWP together with PBS for measurement. The output modes  of  PBS  are coupled to single mode fiber  and  connected to single photon detectors. The prepare and measure \cglmp experiment (B) is realized using a Sagnac configuration. A 50/50 non-polarizing beam splitter (BS) divides photon path in two followed by a combination of two HWP and two QWP together with phase shifter (PP). The measurement part consists of two measurement stations at the two output ports of the BS. Both measurement stations contain a combination of one HWP and one QWP succeeded by a PBS and detected by single mode fiber coupled to single photon detector. See text for details.}\label{setup}
\end{figure*}

\paragraph{Experimental realization}%
Here we will present the experimental implementation of the CHSH and \cglmp contextuality tests. We are using an optical setup where the physical system is defined by single photons representing either a 2-level (Fig.~\ref{setup} (A)) or 4-level (Fig.~\ref{setup} (B)) quantum system. For the CHSH test the 2-level system is encoded in polarization with basis state of $\ket{H}$ and $\ket{V}$, where ($H$)/($V$) is horizontal/vertical polarization photonic mode. We are using polarization and path encoding for the \cglmp test. The four basis states are $\ket{0} = \ket{H,a}$, $\ket{1} = \ket{V,a}$, $\ket{2} = \ket{H,b}$ and $\ket{3} = \ket{V,b}$, where ($H$)/($V$) is horizontal/vertical polarization photonic mode and (a)/(b) is spatial photonic mode of single photons. For both experimental tests, the single photons are generated by heralded single-photon source utilizing a spontaneous parametric down-conversion (SPDC) process where a 2mm non-linear crystal (BBO) is pumped by a femtosecond pulsed laser at an operating wavelength of 390nm. The non-linear process generates twin photons, commonly known as signal and idler photons, which are optically filtered and coupled into single-mode fibers giving a well-defined spectral and spatial characteristic. The idler photon is used as a trigger and is detected by a single detector $D_T$. The signal photons are coupled to experimental setup (A)/(B) as seen in Fig.~\ref{setup}. 

\paragraph{CHSH experimental test}%
The signal photons are initially prepared in $\ket{H}$ state using a polarizer oriented horizontally. The state is prepared by setting the orientation angle of  HWP, see details in Table~\ref{A1} in the appendix. The measurement is performed by choosing the orientation of an HWP followed by  PBS; see details of measurement settings in Table~\ref{A2} in the appendix. The output modes of  PBS are coupled to single-mode fibers and then connected to  Single-photon detectors $D_1$ and $D_2$. The detector output signals are coupled to a multichannel coincidence unit. For each preparation and measurement setting, the measurement was performed for 30 minutes at a rate of 1,500 2-photon coincidences per second. The experimental protocol allowed the classical limit to be violated in accordance with the theory: $2 \leq 2.8021\pm 0.0098 \leq 2/\sqrt{2}$, see details of each term in the inequality in the appendix table~\ref{A6}.

\paragraph{\cglmp experimental test}%
Signal photons are first prepared in $\ket{H}$ state using a polarizer oriented horizontally. We then use a Sagnac configuration in following manner: a 50/50 non-polarizing beam splitter (BS) splits the photon path in two, each containing a combination of two oriented half (HWP) and two quarter-wave plates (QWP), followed by a phase shifter (PP) in one of the paths resulting in the following preparation state: $\frac{1}{\sqrt{2}} (\cos{(2\alpha)} \ket{H,1} + \e^{i\phi_1} \sin{(2\alpha)} \ket{V,1} + \e^{iPP} (\cos{(2\beta)} \ket{H,2} + \e^{i\phi_2} \sin{(2\beta)} \ket{V,2}))$, where $\phi_1$ and $\phi_2$ are the phases acquired in each path relative to horizontal part of the state. Both paths are then again recombined in the BS which before sent to the measuring part of the experiment.

The measurement part consists of two measurement stations at the two output ports of the BS. Both measurement stations contain a combination of HWP and QWP succeeded by a polarizing beam splitter (PBS) and detected by single mode fiber which is coupled to single photon detectors $D_1$, $D_2$, $D_3$ and $D_4$. The detector output signals are connected to a multichannel coincidence unit. Each setting in the \cglmp experiment were conducted with a coincidence counting rate of approximately 30,000 per second for 60 seconds. All the preparation and measurement settings can be found in the Tables \ref{A3}, \ref{A4}, and \ref{A5} in the appendix.

Two \cglmp experiments were performed. First one experimental-friendly using separable states which required fewer rotation setting in preparation stage and the second one to achieve maximal violation. The following protocols were used: first, the separable preparations in \eqref{eq:CGLMP4separableAssemblage} and measurements in \eqref{eq:CGLMP4optimalMeasurement} were used to violate the qutrit bound: $0.315 < 0.3292 \pm 0.008$ (to be compared with the theoretical value of $0.3361$). Then the entangled preparations in \eqref{eq:CGLMP4optimalAssemblage} and measurements \eqref{eq:CGLMP4optimalMeasurement} were used to obtain a higher violation of $0.3631 \pm 0.01$ (to be compared with the theoretical value of $0.3648$). The experimental details of each term in the \cglmp inequality can be found in the Tables~\ref{A7} and \ref{A8} in the Appendix.

\paragraph{Errors}%
The errors in the experiment result are coming from imperfection in state preparation and measurement. On the preparation side, deviation from a perfect state of visibility is due to polarization deviation and imperfections in the spatial overlap in the Sagnac loop, together with rotation errors of the wave plates. At the measurement side, errors arise due to deviations in wave plate rotations and statistical deviations. We estimated experimental errors using Monte Carlo simulation, where the inputs are carefully analyzed by characterizing individual components. The motorized wave plates in experiment (A) have a precision of repeatability of 0.1 degree, which we use as the standard deviation of the Gaussian distribution from which samples are fed into the Monte Carlo simulation. The non-motorized wave plates and phase plate used in experiment (B) are set to have an upper experimental error of 0.5 degree. Finally, we have also considered statistical errors of photon counting, modeled by poissonian statistics. However, statistical errors are due to large experimental sample collection, up to two orders smaller than the errors arising from the Monte Carlo simulation.

\paragraph{Conclusion}%
We have demonstrated that dimension-dependent contextuality inequalities can be used for experimental benchmarks of the number of qubits accessible in a quantum system. Our results made three primary contributions: first, we provided a formal framework for converting Bell inequalities to contextuality inequalities with dimension-dependent bounds; second, we introduced a concrete example based on the \cglmp inequality; and third, we demonstrated the applicability of our theoretical results through a photonic experiment. Moreover the developed methods have significant advantage when compared to previous similar results \cite{brunner2008testing,gallego2010device,brunner2013dimension,bowles2014certifying,tavakoli2015quantum,spee2020genuine} due to the effect that the violation of a contextuality inequality itself can be seen as a signature of non-classicality. Hence the tests we have developed cannot easily be cheated by using larger classical systems and thus they certify quantum degrees of freedom. This means that if our methods would be used to verify the number of qubits in a quantum computer, they would also immediately verify that the device in fact has some non-classical features, indicating that it may work according to the laws of quantum theory. Thus our results not only deepen the understanding of contextuality in quantum systems but also demonstrate its practical usefulness for benchmarking quantum computers.

There are two natural continuations of this research direction: our case of \cglmp inequality can certify the access to two qubits, but current day quantum computers boast up to over 1000 qubits \cite{wilkins2024record, castelvecchi2023ibm}. In order to use our methods to test similar quantum computers, one would need to identify relevant Bell inequalities or contextuality inequalities and prove the dimension bounds. This will be a hard task as already for circa 10 qubits one cannot simply use the same numerical methods as we have used. Another natural continuation of the presented research is to look whether contextuality can be used to certify other application-oriented properties of quantum systems. For example, the connection between contextuality and no-broadcasting was recently observed \cite{jokinen2024no}, hinting at the potential uses of contextuality to benchmark quantum key distribution \cite{wolf2021quantum,rusca2024quantum}.

\begin{acknowledgments}
\paragraph{Acknowledgements}%
This work was supported by the Knut and Alice Wallenberg Foundation through the Wallenberg Centre for Quantum Technology (WACQT) and the Swedish Research Council (VR). MP and OG acknowledge support from the Deutsche Forschungsgemeinschaft (DFG, German Research Foundation, project numbers 447948357 and 440958198), the Sino-German Center for Research Promotion (Project M-0294), the German Ministry of Education and Research (Project QuKuK, BMBF Grant No. 16KIS1618K), the DAAD, and the Alexander von Humboldt Foundation.
\end{acknowledgments}

\bibliography{citations}

\onecolumngrid
\appendix

\section{Proof of Theorem~\ref{thm:PMthenBell}}
\begin{theorem*}
Consider the subnormalized preparations $\{\sigma_{a|x}\}_{a,x}$ obeying
\begin{equation}
\sum_a \sigma_{a|x} = \sigma_*, \forall x
\end{equation}
where $\Tr(\sigma_*) = 1$ and set of POVMs $\{M_{b|y}\}_{b,y}$. Let $p(ab|xy) = \Tr( \sigma_{a|x} M_{b|y})$ to denote the conditional probability one obtains from the experiment, conditioned on the choice of $x$ and $y$. If $(\{\sigma_{a|x}\}_{a,x}, \{M_{b|y}\}_{b,y})$ is P\&M noncontextual, then
\begin{equation}
p(ab|xy) = \sum_\lambda p(\lambda) p(a|x,\lambda) p(b|y,\lambda)
\end{equation}
where $p(\lambda) \geq 0$, $\sum_\lambda p(\lambda) = 1$, $p(a|x,\lambda) \geq 0$, $\sum_a p(a|x,\lambda) = 1$, $p(b|y,\lambda) \geq 0$, $\sum_b p(b|y,\lambda) = 1$.
\end{theorem*}
\begin{proof}
Since $(\{\sigma_{a|x}\}_{a,x}, \{M_{b|y}\}_{b,y})$ is P\&M noncontextual, we have
\begin{equation}
\Tr( \sigma_{a|x} M_{b|y}) = \sum_\lambda \Tr(\sigma_{a|x} N_\lambda) \Tr(\omega_\lambda M_{b|y})
\end{equation}
We will denote
\begin{align}
p(\lambda) &= \sum_a \Tr(\sigma_{a|x} N_\lambda), \\
p(a|x, \lambda) &= \dfrac{\Tr(\sigma_{a|x} N_\lambda)}{p(\lambda)}, \\
p(b|y, \lambda) &= \Tr(\omega_\lambda M_{b|y}).
\end{align}
We have $p(\lambda) \geq 0$, $p(a|x, \lambda) \geq 0$, $p(b|y, \lambda) \geq 0$ as a result of the definition of P\&M noncontextuality. Moreover we have
\begin{equation}
\sum_\lambda p(\lambda) = \sum_{a, \lambda} \Tr(\sigma_{a|x} N_\lambda) = 1
\end{equation}
and so $p(\lambda)$ is a probability distribution. Also
\begin{align}
\sum_a p(a|x, \lambda) &= \dfrac{1}{p(\lambda)} \sum_a \Tr(\sigma_{a|x} N_\lambda) = 1, \\
\sum_b p(b|y, \lambda) &= \sum_b \Tr(\omega_\lambda M_{b|y}) = \Tr(\omega_\lambda) = 1,
\end{align}
and so $p(a|x, \lambda)$, $p(b|y, \lambda)$ are conditional probability distributions. We then have
\begin{equation}
\Tr( \sigma_{a|x} M_{b|y}) = \sum_\lambda \Tr(\sigma_{a|x} N_\lambda) \Tr(\omega_\lambda M_{b|y}) = \sum_\lambda p(\lambda) p(a|x, \lambda) p(b|y, \lambda).
\end{equation}
\end{proof}

\section{Proof of Theorem~\ref{thm:PMdim}}
\begin{theorem*}
Let $W$ be a Bell inequality with dimension-dependent bound $X_d$. Let $\dim(\Ha_d) = d$ and let $\sigma_{a|x} \in \bound(\Ha_d)$ be subnormalized preparations obeying $\sum_a \sigma_{a|x} = \sigma_*, \forall$ where $\Tr(\sigma_*) = 1.$ and let $M_{b|y} \in \bound(\Ha_d)$ be a set of POVMs. Then
\begin{equation}
W([\Tr(\sigma_{a|x} M_{b|y})]) \leq X_d.
\end{equation}
\end{theorem*}
\begin{proof}
GHJW theorem \cite{schrodinger1935discussion,hadjisavvas1981properties,gisin1989stochastic,hughston1993complete,sainz2015postquantum} states that since $\sigma_{a|x}$ is a no-signaling assemblage there is a quantum state $\rho \in \bound(\Ha_d \otimes \Ha_d)$ and set of POVMs $N_{a|x} \in \bound(\Ha_d)$ such that
\begin{equation}
\sigma_{a|x} = \Tr_A (\rho (N_{a|x} \otimes \I))
\end{equation}
where $\Tr_A$ denotes the partial trace over the first subsystem. We thus get
\begin{equation}
\Tr(\sigma_{a|x} M_{b|y}) = \Tr(\rho (N_{a|x} \otimes M_{b|y}))
\end{equation}
and thus
\begin{equation}
W([\Tr(\sigma_{a|x} M_{b|y})]) = W([\Tr(\rho (N_{a|x} \otimes M_{b|y}))]) \leq X_d.
\end{equation}
\end{proof}

\section{See-saw algorithm to maximize the violation of \cglmp with separable preparations}
Let $\Ha = \Ha_2 \otimes \Ha_2$ where $\dim(\Ha_2) = 2$ and $a,b \in \{0, 1, 2, 3\}$ and $x,y \in \{0, 1\}$. If we fix set of POVMs $M_{b|y} \in \bound(\Ha)$, then finding the separable subnormalized preparations $\sigma_{a|x} \in \bound(\Ha)$ that maximize the violation of the \cglmp inequality amounts to solving the following SDP:
\begin{equation} \label{eq:seesaw-assemblage}
\begin{split}
\text{given} \quad &M_{b|y} \\
\text{maximize} \quad & p(a \leq b|00) + p(a \geq b|01)+ p(a \geq b|10) - p(a \geq b|11) - 2 \\
\text{such that} \quad &p(ab|xy) = \Tr(\sigma_{a|x} M_{b|y}) \\
&\sigma_{a|x} \geq 0, \forall a, x \\
&\sigma_{a|x}^{\intercal_1} \geq 0, \forall a, x \\
&\sum_a \sigma_{a|x} = \sigma_*, \forall x \\
&\Tr(\sigma_*) = 1.
\end{split}
\end{equation}
Here $A^{\intercal_1}$ denotes the partial transpose of $A$ over the first subsystem and $p(a \leq b|00)$ is the sum of all of the terms $p(ab|00)$ where $a \leq b$, and analogically for the other terms. Alternately, if we fix the subnormalized preparations $\sigma_{a|x} \in \bound(\Ha)$, then finding the set of POVMs $M_{b|y} \in \bound(\Ha)$ that maximize the violation of the \cglmp inequality amounts to solving a similar SDP:
\begin{equation} \label{eq:seesaw-measurements}
\begin{split}
\text{given} \quad &\sigma_{a|x} \\
\text{maximize} \quad & p(a \leq b|00) + p(a \geq b|01)+ p(a \geq b|10) - p(a \geq b|11) - 2 \\
\text{such that} \quad &p(ab|xy) = \Tr(\sigma_{a|x} M_{b|y}) \\
&M_{b|y} \geq 0, \forall b, y \\
&\sum_b M_{b|y} = \I.
\end{split}
\end{equation}

In order to find the best possible violation of the \cglmp inequality with separable preparations, we do the following: fix random projective measurements $M_{b|y}$ and numerically solve \eqref{eq:seesaw-assemblage} to find the highest violation of the \cglmp inequality with separable subnormalized preparations. Then numerically solve \eqref{eq:seesaw-measurements} to find set of POVMs with potentially higher violation of the \cglmp inequality. \eqref{eq:seesaw-assemblage} and \eqref{eq:seesaw-measurements} are alternately solved until the same result is obtained ten times in a row. Then a different random projective measurements $M_{b|y}$ are generated and the process is repeated.

\section{Experimental settings and results}
\setcounter{figure}{0} \renewcommand{\thefigure}{A.\arabic{figure}} 
\setcounter{table}{0} \renewcommand{\thetable}{A.\arabic{table}}

\begin{table}[h!]
\begin{minipage}[t]{0.475\linewidth}
\caption{\label{A1}Orientation of the HWP for the state preparation for the CHSH experiment (in degrees).}
\begin{tabular}{c|c}
 \hline \hline
 Preparation state & HWP \\ \hline
 $\sigma_{0|0}$& 0 \\
 $\sigma_{0|1}$& 45 \\
 $\sigma_{1|0}$& 22.5 \\
 $\sigma_{1|1}$& -22.5 \\
 \hline \hline
\end{tabular}
\end{minipage}
\hfill
\begin{minipage}[t]{0.475\linewidth}
\caption{\label{A2}Orientation of the HWP for the measurement for the CHSH experiment (in degrees).}
\begin{tabular}{c|c}
 \hline \hline
 Measurement basis & HWP \\ \hline
 $M_{0}$& -78.75 \\
 $M_{1}$& 33.75 \\
  \hline \hline
\end{tabular}
\end{minipage}

\end{table}

\begin{table}[h!]
\caption{\label{A3}Rotation settings of the six out of eight active wave plates (in degrees) and phase shifter setting in "experimental-friendly" realization of \cglmp. Preparation of the separable states in the two preparation basis $M_{b|0}$ respectively $M_{b|1}$.}
\begin{ruledtabular}
\begin{tabular}{c|ccccccccc}
 Prep. & HWP1\textsubscript{a} & QWP1\textsubscript{a} & HWP2\textsubscript{a}& QWP2\textsubscript{a} & HWP1\textsubscript{b} & QWP1\textsubscript{b} & HWP2\textsubscript{b}& QWP2\textsubscript{b} & PP \\ \hline
 $M_{0|0}$&119.22&307.08&-&21.38&188.96&50.84&-&355.50&$\pi/2$ \\
 $M_{1|0}$&99.94&300.34&-&349.54&131.88&358.26&-&212.92&$\pi/2$ \\
 $M_{2|0}$&60.78&232.94&-&158.64&171.04&309.16&-&4.50&$\pi/2$ \\
 $M_{3|0}$&260.06&239.66&-&190.46&103.46&202.42&-&237.08&$\pi/2$ \\
 &&&&&&&&& \\
 $M_{0|1}$&119.22&217.06&-&111.36&98.96&320.84&-&265.5&$0$ \\
 $M_{1|1}$&260.06&329.66&-&280.46&138.12&271.74&-&57.08&$0$ \\
 $M_{2|1}$&60.78&142.94&-&248.64&261.04&219.16&-&94.50&$0$ \\
 $M_{3|1}$&239.14&38.74&-&169.54&346.54&67.58&-&32.92&$0$ \\
\end{tabular}
\end{ruledtabular}
\end{table}

\begin{table}[h!]
\caption{\label{A4}The orientation of the wave plates and phase shifter in realization of \cglmp experiment. Preparation of the states in the two preparation basis $M_{b|0}$ respectively $M_{b|1}$.}
\begin{ruledtabular}
\begin{tabular}{c|ccccccccc}
 Prep. & HPW1\textsubscript{a} & QWP1\textsubscript{a} & HPW2\textsubscript{a}& QWP2\textsubscript{a} & HWP1\textsubscript{b} & QWP1\textsubscript{b} & HPW2\textsubscript{b}& QWP2\textsubscript{b} & PP \\ \hline
 $M_{0|0}$&18.24&45.00&-39.38&45.00&26.76&45.00&-39.38&45.00&$3\pi/4$ \\
 $M_{1|0}$&18.24&45.00&-16.88&45.00&26.76&45.00&-16.88&45.00&$-\pi/4$ \\
 $M_{2|0}$&18.24&45.00&-84.38&45.00&26.76&45.00&-84.38&45.00&$3\pi/4$ \\
 $M_{3|0}$&18.24&45.00&-61.88&45.00&26.76&45.00&-61.88&45.00&$-\pi/4$ \\
 &&&&&&&&& \\
 $M_{0|1}$&18.24&45.00&-50.63&45.00&26.76&45.00&-50.63&45.00&$\pi/4$ \\
 $M_{1|1}$&18.24&45.00&-28.126&45.00&26.76&45.00&-28.126&45.00&$-3\pi/4$ \\
 $M_{2|1}$&18.24&45.00&-5.63&45.00&26.76&45.00&-5.63&45.00&$\pi/4$ \\
 $M_{3|1}$&18.24&45.00&-73.126&45.00&26.76&45.00&-73.126&45.00&$-3\pi/4$ \\
\end{tabular}
\end{ruledtabular}
\end{table}

\begin{table}[h!]
\caption{\label{A5}Orientation of the HWP for the measurement for the \cglmp experiment for measurement basis $\sigma_0$ respectively $\sigma_1$ (in degrees).}
\begin{tabular}{c|cccc}
 \hline \hline
 Measurement basis & QWP\textsubscript{a} & HWP\textsubscript{a} & QWP\textsubscript{b}& HWP\textsubscript{b} \\ \hline
 $\sigma_0$&45.00&22.50&0.00&22.50 \\
 $\sigma_1$&45.00&11.25&-45.00&11.25 \\
  \hline \hline
\end{tabular}
\end{table}

\begin{table}[h!]
\caption{\label{A6}Experimental and theoretical values for each  preparation-measure set of the CHSH experiment. }
\begin{minipage}[t]{0.475\linewidth}
\begin{tabular}{c|cc}
 \hline\hline
   & Expected & Measured \\ \hline
 $M_{1}\sigma_{1|+}$ & 0.3536 & 0.3516 \\
 $M_{1}\sigma_{1|-}$ & -0.3536 & -0.3294 \\
 $M_{1}\sigma_{2|+}$ & 0.3536 & 0.3690 \\
 $M_{1}\sigma_{2|-}$ & -0.3536 & -0.3517 \\
 \hline\hline
\end{tabular}
\end{minipage}
\hfill
\begin{minipage}[t]{0.475\linewidth}
\begin{tabular}{c|cc}
 \hline\hline
   & Expected & Measured \\ \hline
 $M_{2}\sigma_{1|+}$ & 0.3536 & 0.3596 \\
 $M_{2}\sigma_{1|-}$ & -0.3536 & -0.3411 \\
 $M_{2}\sigma_{2|+}$ & -0.3536 & -0.3411 \\
 $M_{2}\sigma_{2|-}$ & 0.3536 & 0.3587 \\
 \hline\hline
\end{tabular}
\end{minipage}

\end{table}

% Experimental values of the "experimental-friendly CGLMP_4
\begin{table}[H]
\caption{\label{A7}Experimental and theoretical probability values of the experimental-friendly \cglmp using separable states. }
\begin{minipage}[t]{0.475\linewidth}
\centering
\begin{tabular}{c|cc}
 \hline\hline
 Set $\sigma_{a|0}M_{b|0}$ & Expected & Measured \\ \hline
 $\Tr(\sigma_{0|0}M_{0|0})$ & 0.8211 & 0.8136 \\
 $\Tr(\sigma_{0|0}M_{1|0})$ & 0.0453 & 0.0172 \\
 $\Tr(\sigma_{1|0}M_{1|0})$ & 0.8211 & 0.8398 \\
 $\Tr(\sigma_{0|0}M_{2|0})$ & 0.0324 & 0.0189 \\
 $\Tr(\sigma_{1|0}M_{2|0})$ & 0.0453 & 0.0319 \\
 $\Tr(\sigma_{2|0}M_{2|0})$ & 0.8211 & 0.8540 \\
 $\Tr(\sigma_{0|0}M_{3|0})$ & 0.1012 & 0.0967 \\
 $\Tr(\sigma_{1|0}M_{3|0})$ & 0.0324 & 0.0908 \\
 $\Tr(\sigma_{2|0}M_{3|0})$ & 0.0453 & 0.0455 \\
 $\Tr(\sigma_{3|0}M_{3|0})$ & 0.8211 & 0.7670 \\
 \hline\hline
\end{tabular}
\end{minipage}
\hfill
\begin{minipage}[t]{0.475\linewidth}
\centering
\begin{tabular}{c|cc}
 \hline\hline
 Set $\sigma_{a|1}M_{b|0}$ & Expected & Measured \\ \hline
 $\Tr(\sigma_{0|1}M_{0|0})$ & 0.8211 & 0.7845 \\
 $\Tr(\sigma_{1|1}M_{0|0})$ & 0.0453 & 0.0546 \\
 $\Tr(\sigma_{2|1}M_{0|0})$ & 0.0324 & 0.0735 \\
 $\Tr(\sigma_{3|1}M_{0|0})$ & 0.1012 & 0.0874 \\
 $\Tr(\sigma_{1|1}M_{1|0})$ & 0.8211 & 0.9223 \\
 $\Tr(\sigma_{2|1}M_{1|0})$ & 0.0453 & 0.0034 \\
 $\Tr(\sigma_{3|1}M_{1|0})$ & 0.0324 & 0.0178 \\
 $\Tr(\sigma_{2|1}M_{2|0})$ & 0.8211 & 0.7330 \\
 $\Tr(\sigma_{3|1}M_{3|0})$ & 0.0453 & 0.0818 \\
 $\Tr(\sigma_{3|1}M_{3|0})$ & 0.8211 & 0.8027 \\
 \hline\hline
\end{tabular}
\end{minipage}

\vspace{1em}
\begin{minipage}[t]{0.475\linewidth}
\centering
\begin{tabular}{c|cc}
 \hline\hline
 Set $\sigma_{a|0}M_{b|1}$ & Expected & Measured \\ \hline
 $\Tr(\sigma_{0|0}M_{0|1})$ & 0.8211 & 0.7342 \\
 $\Tr(\sigma_{1|0}M_{0|1})$ & 0.0453 & 0.0962 \\
 $\Tr(\sigma_{2|0}M_{0|1})$ & 0.0324 & 0.0662 \\
 $\Tr(\sigma_{3|0}M_{0|1})$ & 0.1012 & 0.1034 \\
 $\Tr(\sigma_{1|0}M_{1|1})$ & 0.8211 & 0.8387 \\
 $\Tr(\sigma_{2|0}M_{1|1})$ & 0.0453 & 0.0690 \\
 $\Tr(\sigma_{3|0}M_{1|1})$ & 0.0324 & 0.0087 \\
 $\Tr(\sigma_{2|0}M_{2|1})$ & 0.8211 & 0.7635 \\
 $\Tr(\sigma_{3|0}M_{2|1})$ & 0.0453 & 0.0478 \\
 $\Tr(\sigma_{3|0}M_{3|1})$ & 0.8211 & 0.8729 \\
 \hline\hline
\end{tabular}
\end{minipage}
\hfill
\begin{minipage}[t]{0.475\linewidth}
\centering
\begin{tabular}{c|cc}
 \hline\hline
 Set $\sigma_{a|1}M_{b|1}$ & Expected & Measured \\ \hline
 $\Tr(\sigma_{0|1}M_{0|1})$ & 0.1012 & 0.0727 \\
 $\Tr(\sigma_{1|1}M_{0|1})$ & 0.0324 & 0.0853 \\
 $\Tr(\sigma_{2|1}M_{0|1})$ & 0.0453 & 0.0496 \\
 $\Tr(\sigma_{3|1}M_{0|1})$ & 0.8211 & 0.7925 \\
 $\Tr(\sigma_{1|1}M_{1|1})$ & 0.1012 & 0.1017 \\
 $\Tr(\sigma_{2|1}M_{1|1})$ & 0.0324 & 0.0312 \\
 $\Tr(\sigma_{3|1}M_{1|1})$ & 0.0453 & 0.0625 \\
 $\Tr(\sigma_{2|1}M_{2|1})$ & 0.1012 & 0.1028 \\
 $\Tr(\sigma_{3|1}M_{2|1})$ & 0.0324 & 0.0649 \\
 $\Tr(\sigma_{3|3}M_{3|1})$ & 0.1012 & 0.0772 \\
 \hline\hline
\end{tabular}
\end{minipage}

\end{table}

The conditional probabilities in \ref{A7} are summed up to:
\begin{align}
\begin{split}
p(a \leq b|00) &=
\Tr(\sigma_{0|0}M_{0|0}) + \Tr(\sigma_{0|0}M_{1|0}) + \Tr(\sigma_{1|0}M_{1|0}) + \Tr(\sigma_{0|0}M_{2|0}) + \Tr(\sigma_{1|0}M_{2|0}) \\
&+ \Tr(\sigma_{2|0}M_{2|0}) + \Tr(\sigma_{0|0}M_{3|0}) + \Tr(\sigma_{1|0}M_{3|0}) + \Tr(\sigma_{2|0}M_{3|0}) + \Tr(\sigma_{3|0}M_{3|0}) = 0.8939
\end{split} \\
\begin{split}
p(a \geq b|10) &=
\Tr(\sigma_{0|1}M_{0|0}) + \Tr(\sigma_{1|1}M_{0|0}) + \Tr(\sigma_{2|1}M_{0|0}) + \Tr(\sigma_{3|1}M_{0|0}) + \Tr(\sigma_{1|1}M_{1|0}) \\
&+ \Tr(\sigma_{2|1}M_{1|0}) + \Tr(\sigma_{3|1}M_{1|0}) + \Tr(\sigma_{2|1}M_{2|0}) + \Tr(\sigma_{3|1}M_{2|0}) + \Tr(\sigma_{3|1}M_{3|0}) = 0.8952
\end{split} \\
\begin{split}
p(a \geq b|01) &=
\Tr(\sigma_{0|0}M_{0|1}) + \Tr(\sigma_{1|0}M_{0|1}) + \Tr(\sigma_{2|0}M_{0|1}) + \Tr(\sigma_{3|0}M_{0|1}) + \Tr(\sigma_{1|0}M_{1|1}) \\
&+ \Tr(\sigma_{2|0}M_{1|1}) + \Tr(\sigma_{3|0}M_{1|1}) + \Tr(\sigma_{2|0}M_{2|1}) + \Tr(\sigma_{3|0}M_{2|1}) + \Tr(\sigma_{3|0}M_{3|1}) = 0.9001
\end{split} \\
\begin{split}
p(a \geq b|11) &=
\Tr(\sigma_{0|1}M_{0|1}) + \Tr(\sigma_{1|1}M_{0|1}) + \Tr(\sigma_{2|1}M_{0|1}) + \Tr(\sigma_{3|1}M_{0|1}) + \Tr(\sigma_{1|1}M_{1|1}) \\
&+ \Tr(\sigma_{2|1}M_{1|1}) + \Tr(\sigma_{3|1}M_{1|1}) + \Tr(\sigma_{2|1}M_{2|1}) + \Tr(\sigma_{3|1}M_{2|1}) + \Tr(\sigma_{3|1}M_{3|1}) = 0.3601
\end{split}
\end{align}
Which results in the experimental-friendly \cglmp violation:
\begin{equation}
p(a \leq b|00) + p(a \geq b|01)+ p(a \geq b|10) - p(a \geq b|11) - 2 = 0.3292.
\end{equation}

% Experimental values of the  CGLMP_4
\begin{table}[H]
\caption{\label{A8}Experimental and theoretical probability values of the \cglmp experiment.}
\begin{minipage}[t]{0.475\linewidth}
\centering
\begin{tabular}{c|cc}
 \hline\hline
 Set $\sigma_{a|0}M_{b|0}$ & Expected & Measured \\ \hline
 $\Tr(\sigma_{0|0}M_{0|0})$ & 0.7833 & 0.8158 \\
 $\Tr(\sigma_{0|0}M_{1|0})$ & 0.1050 & 0.1035 \\
 $\Tr(\sigma_{1|0}M_{1|0})$ & 0.7833 & 0.7817 \\
 $\Tr(\sigma_{0|0}M_{2|0})$ & 0.0547 & 0.0589 \\
 $\Tr(\sigma_{1|0}M_{2|0})$ & 0.1050 & 0.1006 \\
 $\Tr(\sigma_{2|0}M_{2|0})$ & 0.7833 & 0.7861 \\
 $\Tr(\sigma_{0|0}M_{3|0})$ & 0.0569 & 0.0380 \\
 $\Tr(\sigma_{1|0}M_{3|0})$ & 0.0547 & 0.0416 \\
 $\Tr(\sigma_{2|0}M_{3|0})$ & 0.1050 & 0.1129 \\
 $\Tr(\sigma_{3|0}M_{3|0})$ & 0.7833 & 0.8074 \\
 \hline\hline
\end{tabular}
\end{minipage}
\hfill
\begin{minipage}[t]{0.475\linewidth}
\centering
\begin{tabular}{c|cc}
 \hline\hline
 Set $\sigma_{a|1}M_{b|0}$ & Expected & Measured \\ \hline
 $\Tr(\sigma_{0|1}M_{0|0})$ & 0.7833 & 0.8113 \\
 $\Tr(\sigma_{1|1}M_{0|0})$ & 0.1050 & 0.1122 \\
 $\Tr(\sigma_{2|1}M_{0|0})$ & 0.0547 & 0.0554 \\
 $\Tr(\sigma_{3|1}M_{0|0}0$ & 0.0569 & 0.0211 \\
 $\Tr(\sigma_{1|1}M_{1|0})$ & 0.7833 & 0.7969 \\
 $\Tr(\sigma_{2|1}M_{1|0})$ & 0.1050 & 0.1281 \\
 $\Tr(\sigma_{3|1}M_{1|0})$ & 0.0547 & 0.0301 \\
 $\Tr(\sigma_{2|1}M_{2|0})$ & 0.7833 & 0.8597 \\
 $\Tr(\sigma_{3|1}M_{3|0})$ & 0.0569 & 0.0211 \\
 $\Tr(\sigma_{3|1}M_{3|0})$ & 0.7833 & 0.7833 \\
 \hline\hline
\end{tabular}
\end{minipage}

\vspace{1em}
\begin{minipage}[t]{0.475\linewidth}
\centering
\begin{tabular}{c|cc}
 \hline\hline
 Set $\sigma_{a|0}M_{b|1}$ & Expected & Measured \\ \hline
 $\Tr(\sigma_{0|0}M_{0|1})$ & 0.7833 & 0.8013 \\
 $\Tr(\sigma_{1|0}M_{0|1})$ & 0.1050 & 0.1277 \\
 $\Tr(\sigma_{2|0}M_{0|1})$ & 0.0547 & 0.0500 \\
 $\Tr(\sigma_{3|0}M_{0|1})$ & 0.0569 & 0.0210 \\
 $\Tr(\sigma_{1|0}M_{1|1})$ & 0.7833 & 0.7607 \\
 $\Tr(\sigma_{2|0}M_{1|1})$ & 0.1050 & 0.1275 \\
 $\Tr(\sigma_{3|0}M_{1|1})$ & 0.0547 & 0.0557 \\
 $\Tr(\sigma_{2|0}M_{2|1})$ & 0.7833 & 0.8030 \\
 $\Tr(\sigma_{3|0}M_{2|1})$ & 0.1050 & 0.0986 \\
 $\Tr(\sigma_{3|0}M_{3|1})$ & 0.7833 & 0.7712 \\
 \hline\hline
\end{tabular}
\end{minipage}
\hfill
\begin{minipage}[t]{0.475\linewidth}
\centering
\begin{tabular}{c|cc}
 \hline\hline
 Set $\sigma_{a|1}M_{b|1}$ & Expected & Measured \\ \hline
 $\Tr(\sigma_{0|1}M_{0|1})$ & 0.0569 & 0.0222 \\
 $\Tr(\sigma_{1|1}M_{0|1})$ & 0.0547 & 0.0546 \\
 $\Tr(\sigma_{2|1}M_{0|1})$ & 0.1050 & 0.0999 \\
 $\Tr(\sigma_{3|1}M_{0|1})$ & 0.7833 & 0.8232 \\
 $\Tr(\sigma_{1|1}M_{1|1})$ & 0.0569 & 0.0437 \\
 $\Tr(\sigma_{2|1}M_{1|1})$ & 0.0547 & 0.0750 \\
 $\Tr(\sigma_{3|1}M_{1|1})$ & 0.1050 & 0.1255 \\
 $\Tr(\sigma_{2|1}M_{2|1})$ & 0.0569 & 0.0727 \\
 $\Tr(\sigma_{3|1}M_{2|1})$ & 0.0547 & 0.0596 \\
 $\Tr(\sigma_{3|3}M_{3|1})$ & 0.0569 & 0.0538 \\
 \hline\hline
\end{tabular}
\end{minipage}

\end{table}

The conditional probabilities in \ref{A8} are summed up to:
\begin{align}
\begin{split}
p(a \leq b|00) &=
\Tr(\sigma_{0|0}M_{0|0}) + \Tr(\sigma_{0|0}M_{1|0}) + \Tr(\sigma_{1|0}M_{1|0}) + \Tr(\sigma_{0|0}M_{2|0}) + \Tr(\sigma_{1|0}M_{2|0}) \\
&+ \Tr(\sigma_{2|0}M_{2|0}) + \Tr(\sigma_{0|0}M_{3|0}) + \Tr(\sigma_{1|0}M_{3|0}) + \Tr(\sigma_{2|0}M_{3|0}) + \Tr(\sigma_{3|0}M_{3|0}) = 0.9116
\end{split} \\
\begin{split}
p(a \geq b|10) &=
\Tr(\sigma_{0|1}M_{0|0}) + \Tr(\sigma_{1|1}M_{0|0}) + \Tr(\sigma_{2|1}M_{0|0}) + \Tr(\sigma_{3|1}M_{0|0}) + \Tr(\sigma_{1|1}M_{1|0}) \\
&+ \Tr(\sigma_{2|1}M_{1|0}) + \Tr(\sigma_{3|1}M_{1|0}) + \Tr(\sigma_{2|1}M_{2|0}) + \Tr(\sigma_{3|1}M_{2|0}) + \Tr(\sigma_{3|1}M_{3|0}) = 0.9042
\end{split} \\
\begin{split}
p(a \geq b|01) &=
\Tr(\sigma_{0|0}M_{0|1}) + \Tr(\sigma_{1|0}M_{0|1}) + \Tr(\sigma_{2|0}M_{0|1}) + \Tr(\sigma_{3|0}M_{0|1}) + \Tr(\sigma_{1|0}M_{1|1}) \\
&+ \Tr(\sigma_{2|0}M_{1|1}) + \Tr(\sigma_{3|0}M_{1|1}) + \Tr(\sigma_{2|0}M_{2|1}) + \Tr(\sigma_{3|0}M_{2|1}) + \Tr(\sigma_{3|0}M_{3|1}) = 0.9048
\end{split} \\
\begin{split}
p(a \geq b|11) &=
\Tr(\sigma_{0|1}M_{0|1}) + \Tr(\sigma_{1|1}M_{0|1}) + \Tr(\sigma_{2|1}M_{0|1}) + \Tr(\sigma_{3|1}M_{0|1}) + \Tr(\sigma_{1|1}M_{1|1}) \\
&+ \Tr(\sigma_{2|1}M_{1|1}) + \Tr(\sigma_{3|1}M_{1|1}) + \Tr(\sigma_{2|1}M_{2|1}) + \Tr(\sigma_{3|1}M_{2|1}) + \Tr(\sigma_{3|1}M_{3|1}) = 0.3576
\end{split}
\end{align}
Which results in the experimental-friendly \cglmp violation:
\begin{equation}
p(a \leq b|00) + p(a \geq b|01)+ p(a \geq b|10) - p(a \geq b|11) - 2 = 0.3631.
\end{equation}

\end{document}